\documentclass[12pt]{article}
\usepackage{preamble}

\def\GKM{\mathsf{GKM}}
\def\AND{\mathsf{AND}}
\def\LONG{\mathsf{long}}
\def\SHORT{\mathsf{short}}
\def\width{w}
\def\terms{\ell}
\def\epsilon{\varepsilon}
\def\e{\eps}
\def\op{\mathsf{op}}
\def\Fr{\mathsf{F}}

\title{Pseudorandom bits for non-commutative programs}

\author{Chin Ho Lee%
\thanks{Work done in part at Harvard University, supported by Madhu Sudan’s and Salil Vadhan’s Simons Investigator Awards.}\\
North Carolina State University%
\and
Emanuele Viola%
\thanks{Supported by NSF grants CCF-2114116 and CCF-2430026.}\\
Northeastern University%
}

\begin{document}
\maketitle

\begin{abstract}
We obtain new explicit pseudorandom generators for several computational models
involving groups.  Our main results are as follows:

\begin{enumerate}
 \item We consider read-once group-products over a finite group $G$,
   i.e., tests of the form $\prod_{i=1}^n g_{i}^{x_{i}}$ where $g_{i}\in G$,
a special case of read-once permutation branching programs. We give
generators with optimal seed length $c_{G}\log(n/\varepsilon)$ over any $p$-group.
The proof uses the small-bias plus noise paradigm, but derandomizes the noise to avoid the recursion in previous work.  Our generator works when the bits are read in any order.  Previously for any non-commutative group the best seed length was
$\ge\log n \log(1/\varepsilon)$, even for a fixed order.

  \item We give a reduction that ``lifts'' suitable generators for
group products over $G$ to a generator that fools \emph{width-$w$ block
products}, i.e., tests of the form $\prod g_{i}^{f_i}$ where the $f_i$ are arbitrary functions on disjoint blocks of $w$ bits.  Block products generalize several previously studied classes.  The reduction applies
to groups that are mixing in a representation-theoretic sense that
we identify.

  \item Combining (2) with (1) and other works we obtain new generators
for block products over the quaternions or over any commutative group,
with nearly optimal seed length. In particular, we obtain generators
for read-once polynomials modulo any fixed $m$ with nearly optimal seed
length. Previously this was known only for $m=2$.

  \item We give a new generator for products over ``mixing groups.''  The construction departs from previous work and uses representation theory.  For constant error, we obtain optimal seed length, improving on previous work (which applied to any group).
\end{enumerate}

This paper identifies a challenge in the area that is
reminiscent of a roadblock in circuit complexity -- handling composite
moduli -- and points to several classes of groups to be attacked next.
\end{abstract}

\section{Introduction}

The construction of explicit pseudorandom generators is a fundamental research
goal that has applications in many areas of theoretical computer science.
For background we refer to the recent survey \cite{HatamiH23}.
We first define pseudorandom generators, incorporating the variants
of \emph{any order} (reflected in the permutation $\pi$) and \emph{non-Boolean} tests (reflected in the
range set $R$).

\begin{definition}[Pseudorandom generators (PRGs)]
\label{def:PRG}An explicit function $P\colon\zo^{s}\to\zo^{n}$ is a \emph{pseudorandom
generator (PRG)} with seed length $s$ and error $\e$ for a class
of functions $F$ mapping $\zo^{n}$ to a set $R$ if for every $f\in F$ the
statistical distance between $f(P(U_{s}))$ and $f(U_{n})$ is $\le\e$,
where $U_{s}$ denotes the uniform distribution over $\zo^{s}$.  We say $P$ fools $F$ \emph{in any order} if $\pi(P)$ fools $F$
for any permutation $\pi$ of the positions of the $n$ input bits.
A PRG is \emph{explicit} if it is computable in time $n^c$.
\end{definition}

\paragraph{PRGs for branching programs, and group programs.}

A main agenda is obtaining explicit pseudorandom generators for read-once branching
programs (ROBPs), with an ultimate goal of proving $\text{BPL}=\text{L}$.
However, even for \emph{constant-width, permutation }ROBPs, the best known
seed length is $\ge\log n \log(1/\e)$. This is $\ge\log^{2}n$
when $\e=1/n$, and thus falls short of the optimal seed length $c\log(n/\e)$.
For permutation ROBPs of width $w$, seed length $c_{w}\log(n/\e)\log(\e^{-1}\log n)$
follows from instantiating the ``Polarizing Random Walks'' \cite{DBLP:journals/toc/ChattopadhyayHH19}
with a bound from \cite{ReingoldSV13,LPV22}. These generators work in \emph{any
order}; thus they essentially match the seed length $c_{w}\log(n/\e)\log(1/\e)$
that was already available for \emph{fixed-order }in a sequence of
exciting works culminating in \cite{DBLP:journals/eccc/Steinke12}.

The class of permutation ROBPs is equivalent to \emph{group
programs} (see e.g.~\cite{DBLP:conf/stoc/KouckyNP11}):
\begin{definition}
A \emph{program} (or \emph{product}) $p$ of length $n$ over a group $G$ is a tuple
$(g_1,g_2,\ldots,g_n)\in G^n$. The program computes the
function $f_p \colon \zo^n \ni x \mapsto \prod_{i\in[n]} g_i^{x_i} \in G$.
\end{definition}

No generator with seed length less than $\log n \log(1/\e)$ was available
for any non-commutative group. While optimal seed length $c\log(n/\e)$
was known for $\mathbb{Z}_{2}$ since \cite{NaN90}, it took nearly
20 years and different techniques to have the same seed length over
$\mathbb{Z}_{3}$ \cite{LovettRTV09,MekaZ09}, and remarkably that
seed length is still not available even for $\mathbb{Z}_{6}$ (see \cite{GKM18} for the best known construction).

\paragraph{PRGs for read-once polynomials.}

Another model that has received significant attention is \emph{read-once
polynomials}. Intuitively, this model can serve as a bridge between
permutation and non-permutation ROBPs. The available generators for
non-permutation ROBPs have significantly worse seed length than for
permutation programs, see e.g.~\cite{DBLP:conf/stoc/MekaRT19} and the discussion
there.

A sequence of works \cite{LeeV-rop,DBLP:conf/stoc/MekaRT19,DBLP:conf/coco/DoronHH20}
culminated in PRGs with seed length $c\log n+\log(1/\e)\log^{c}\log(1/\e)$
for read-once polynomials over $\mathbb{Z}_{2}$. But for other domains
such as $\mathbb{Z}_{3}$ such good seed lengths were not known.

\paragraph{PRGs for block-products.}

A more general model that generalizes and unifies the previous ones
is what we call block-products of width $w$ over a group $G$. Here,
the input bits are arbitrarily partitioned in blocks of $w$ bits,
arbitrary Boolean functions are then applied to each
block, and finally the outputs are used as exponents to group elements.  For our results, we will need to allow one block to be
larger; we call this \emph{spill} and incorporate it in the following
definition.
\begin{definition}[Block-product with spill] \label{def:block-prod-with-spill}
A function $f\colon\zo^n\to G$ is computable by a $w$-block product
with $\ell$ terms and a spill of $q$ bits, written as \emph{$(\ell,w,q)$-product,}
over a group $G$ if there exist $\ell+1$
disjoint subsets $I_{0},I_{1},\ldots,I_{\ell} \subseteq [n]$, where $\abs{I_0}\le q$ and
$\abs{I_i} \le w$ for each $i \in [\ell]$ such that 
\[
  f(x)= \prod_{i=1}^{\ell}g_{i}^{f_{i}(x_{I_{i}})} 
\]
for some group elements $g_{i}\in G$, functions $f_{i}\colon\zo^{I_{i}}\to\zo$.
Here $x_{I_{i}}$ are the $|I_{i}|$ bits of $x$ indexed by $I_{i}$.
\end{definition}

Note that block products are unordered by definition. They are a generalization
of several function classes that have been studied, including \emph{modular
sums}~\cite{LRTV09,MZ09,GKM18} (when $G$ is a cyclic group and
$w=1$), product tests with outputs in $\{-1,1\}$ (a.k.a.~combinatorial
checkerboards) \cite{Watson13,HLV18,LV18,LV20,Lee19} (when $G=\mathbb{Z}_{2}$),
themselves a generalization of \emph{combinatorial rectangles }\cite{ArmoniSWZ96,Lu02,GopalanMRTV12},
and unordered \emph{combinatorial shapes}~\cite{GMRZ13,GKM18} (when
$G=\mathbb{Z}_{m+1}$). Block products also generalize read-once polynomials
because one can show (for the uniform and typically also pseudorandom
distributions) that monomials of degree $\ge\log(n/\e)$ do not affect
the result significantly, and so one can simulate these polynomials
with blocks of size $\log(n/\eps)$.

In terms of generators, a series of works culminating in \cite{Lee19}
gives nearly-optimal seed length (i.e., $w+\log(\ell/\e)$ up to lower-order
factors) over $\mathbb{Z}_{2}$. But such a result was not known over other groups such as $\mathbb{Z}_{3}$ or any non-commutative group.

\subsection{Our results}

In this work we bring new techniques, notably from group theory, to
bear on these problems, and use them to obtain new pseudorandom generators.

First, we obtain optimal seed length for products over $p$-groups.
\begin{definition}
A finite $p$-group is a group of order $p^k$ for an integer $k$
and a prime $p$.
\end{definition}

Equivalently, the order of every element is a power of $p$. (The
latter definition makes sense for infinite groups, but we only consider
finite groups.) The class of $p$-groups is rich and has been studied in
various areas of theory of computation. For example, $p$-groups remain
a candidate for good group-theoretic algorithms for matrix multiplication
\cite{DBLP:journals/corr/abs-1712-02302}; the isomorphism testing
for a subclass of $p$-groups has been identified as a barrier to
faster group isomorphism algorithms \cite{DBLP:conf/stoc/Sun23};
$p$-groups (specifically, unitriangular groups) are used for cryptography
in NC$^{0}$ \cite{ApplebaumIK06} (see \cite{viola-FTTCS09} for
an exposition emphasizing these groups); finally, $p$-groups (specifically,
quaternions) are used in computer graphics to express 3D rotations~\cite{kuipers2002quaternions}.

We now give a few examples of such groups, all of which are non-commutative.
\begin{itemize}
  \item The \emph{quaternion group $\mathbb{Q}_{8}$ }of order $8$ is a $2$-group.

  \item \emph{Unitriangular groups} over $\mathbb{F}_{p}$ are $p$-groups.
They consist of upper-triangular matrices (of some fixed dimension),
with $1$ on the diagonal and entries in $\mathbb{F}_{p}$.

  \item \emph{Wreath products} give natural examples of $p$-groups. For example,
the wreath product $\mathbb{Z}_{p} \wr \mathbb{Z}_{p}$ is a group of
order $p^{p+1}$, hence a $p$-group. This group is the direct product
$\mathbb{Z}_{p}^{p}$ with another element in $\mathbb{Z}_{p}$ acting
on the tuple by shifting the coordinates. For concreteness, the case
$p=2$ can be presented as $(a,b;z)$ where $a,b,z\in\mathbb{Z}_{2}$,
$(a,b;0)(a',b';z')=(a+a',b+b';z')$, and $(a,b;1)(a',b';z')=(a+b',b+a';1+z')$.
Wreath product constructions (not necessarily $p$-groups) have been
studied in a variety of contexts ranging from group-theoretic algorithms
for matrix multiplication \cite{CohnKSU05}, to construction of expander
graphs \cite{AlonLW01,DBLP:journals/toc/RozenmanSW06}, to mixing
in non-quasirandom groups \cite{GowersV-group-mix}.

  \item The \emph{dihedral group }$\mathbb{D}_{n}$ is the group of order $2n$ of symmetries
of a regular polygon with $n$ sides. When $n=2^{t}$, $\mathbb{D}_{n}$ is a $2$-group.
\end{itemize}

We give pseudorandom generators for programs over $p$-groups, with
optimal seed length.  Throughout this paper, we use $c_x$ to denote a constant that depends on the variable $x$.
\begin{theorem}
\label{thm:fooling-p-groups}Let $G$ be a $p$-group. There is an
explicit pseudorandom generator that fools programs of length $n$
over $G$ in any order, with seed length $c_{G}\log(n/\e)$.
\end{theorem}

In fact, the same result holds even for block-products over $p$-groups
with constant block length $w$.

\paragraph{Polynomials and block-products.}
We give a general reduction that ``lifts'' a PRG $P$ for group
products over $G$ to a PRG $P'$ for block-products (and read-once
polynomials)  over $G$.  The reduction applies to any group $G$ that is \emph{mixing}:

\begin{definition}[Mixing groups] \label{def:mixing}
A (finite) group $G$ is \emph{mixing} if for every nontrivial irreducible (unitary) representation $\rho$ and non-identity element $g \in G$, the matrix $\rho(g)$ has no eigenvalue 1.
\end{definition}

\begin{remark}
Our results for mixing groups (\Cref{thm:reduction,thm:fooling-mixing-groups}) apply more generally to fooling \emph{words} over a mixing subset $H$ of a (not necessarily mixing) group $G$.  The property we need is that \Cref{def:mixing} holds for every element in $H$.  There are many examples of mixing subsets of non-mixing groups which generate the entire group $G$.  For example, for $\mathbb{S}_3 = \mathbb{D}_3$, it suffices to exclude the ``flip'' element, i.e. the non-identity element $r$, where $r^2 = 1$.  Moreover, one can have natural examples for infinite groups.  However for simplicity we focus on finite mixing groups.
\end{remark}

We note that mixing groups are exactly the class of \emph{Dedekind groups}.

\begin{definition}
(Finite) \emph{Dedekind groups} are groups of the form $\mathbb{Q}_{8}\times\mathbb{Z}_{2}^{t} \times D$
for any integer $t$ and commutative group $D$ of odd order.  A non-commutative Dedekind group is also called a Hamiltonian group.
\end{definition}

\begin{lemma}[Mixing characterization of Dedekind groups]\label{lem:dedekind-mixing}
  A finite group is mixing if and only if it is Dedekind.
\end{lemma}

A proof of \Cref{lem:dedekind-mixing} is in \Cref{sec:proof-dedeking-mixing}.
We can now state our reduction:
\begin{theorem} \label{thm:reduction}
  Let $G$ be a mixing group. Suppose there is
a PRG $P_{1}$ with seed length $s_{1}$ that $\eps$-fools $(\ell,1,3\log(1/\eps))$-products
over $G$. Then there is a PRG that $\bigl(c_G\log(w + \log(\ell/\eps)) \cdot \eps \bigr)$-fools $(\ell,w,\log(1/\eps))$-products
over $G$, with seed length 
\[
c_{G}\bigl(s_{1}+\log(\ell/\eps)+w\bigr)\cdot\log^{c}\bigl(w+\log(n\ell /\eps) \bigr).
\]
\end{theorem}

Note that if $P_{1}$ has nearly optimal seed length (i.e., $\log(\ell/\e)$
times lower-order terms) then also the final PRG has nearly optimal
seed length (i.e., $w+\log(\ell/\e)$, times lower-order terms).

Applying the reduction (\Cref{thm:reduction}) we obtain near-optimal
PRGs for block products over commutative or Dedekind $2$-groups (in
particular, the quaternions).

\begin{corollary}
\label{cor:intro} Let $G$ be either a commutative group, or a Dedekind
2-group, that is, $G=\mathbb{Q}_{8}\times\mathbb{Z}_{2}^{t}$ for
some $t$. There is an explicit PRG that $\eps$-fools $(\ell,w,0)$-block products over $G$
with seed length $c_{G}(w+\log(\ell/\e))\log^{c}(w+\log(\ell n/\e))$.
\end{corollary}

\begin{proof}
We use the reduction (\Cref{thm:reduction}). For commutative groups
we use the PRG in \cite{GopalanKM15} for $P_1$; for Dedekind
2-groups we use our \Cref{thm:fooling-p-groups} for $P_{1}$. Actually,
in both cases the generators were only stated for group products while
we need to handle the spill. The simple modification is in \Cref{sec:prg-spill-product}.
\end{proof}
As remarked earlier, as a consequence of \Cref{cor:intro}, we obtain
PRGs for read-once polynomials over $n$ variables over any finite field 
$\F$ with near-optimal seed length $c_{\F}\log(n/\e)\log^{c}\log(n/\e$).
Again, this was not known even for $\F_3$.

This result is also a step towards handling group programs over more
general groups, for example \emph{nilpotent groups,} which are direct
products of $p$-groups (for different $p$).  Jumping ahead, our techniques imply that generators for such groups follow from generators for (non-read-once) polynomials over composites.

Finally, we give a new generator for products over mixing groups.
\begin{theorem} \label{thm:fooling-mixing-groups}
Let $G$ be a mixing group. There is an explicit PRG $P$ that $\eps$-fools length-$n$ programs
over $G$ with seed length $c_G \log(n/\e)\log(1/\e)$, in any order.
\end{theorem}

The parameter improvement over previous work appears tiny: As remarked
earlier, \cite{DBLP:journals/toc/ChattopadhyayHH19} gives seed length
$c_{G}\log(n/\e)\log(\e^{-1}\log n)$, and moreover for any $G$.
Still, for constant error we obtain optimal seed length which was
known only in the fixed-order case (cf.~\cite{DBLP:journals/eccc/Steinke12}).
Also note that mixing groups of the form $\mathbb{Q}_{8}\times\mathbb{Z}_{2}^{t}$
(i.e., $m=1$) are $2$-groups, for which we gave optimal seed length
in \Cref{thm:fooling-p-groups}. But the techniques there do not even
apply to the commutative (mixing) group $\mathbb{Z}_{2}\times\mathbb{Z}_{3}$.

Our main interest in this result is that its proof is different from
previous work: it showcases how we can use information on the representation
theory to improve the parameters, pointing to several open problems.

\subsection{Future directions and open problems}

This work suggests that the difficulty of handling more general classes
of groups lies in composite moduli. For example, we do not have new
generators for $\mathbb{D}_{3}= \mathbb{S}_{3}$, a group of order $6$, even though
we have optimal seed length for $\mathbb{D}_{n}$ when $n$ is a power of two.  Thus, a challenge  emerging from this work is to improve the seed length over any non-commutative group of composite order.  Again, $\mathbb{S}_{3}$ is an obvious candidate, which is equivalent to fooling
width-3 permutation ROBPs. But other groups could be easier to handle,
for example Dedekind groups or the direct product
of a $p$-group and a $q$-group where $p\ne q$ are primes.

Also, the techniques in this paper point to several other questions.
Can we extend our reduction to block products where instead of $g^{f}$ for Boolean $f$ we more generally have $g^f$ replaced by a function with range $G$?  For what other groups can we exploit representation theory to obtain better PRGs?

\section{Proof of \Cref{thm:fooling-p-groups}}

We use the fact that programs over $p$-groups can be written as polynomials.
Elements in a group of order $p^{k}$ will be written as $k$-tuples
over $\mathbb{F}_{p}$.

\begin{lemma} \label{lem:p-group-word-polynomial}Let $G$ be a group of order $p^{k}$,
and $k$ an integer. There is a 1-1 correspondence between $G$ and $\F_p^k$ and a polynomial map $f=(f_1,\ldots,f_k) \colon (\F_p^k)^n \to \F_p^k$
over $\mathbb{F}_{p}$ where the $f_{i}\colon (\F_p^k)^n \to \F_p$ have degree $c_{G}$
such that for any $\overline{g}:=(g_1,g_1,\ldots,g_n) \in G^n$ and $x \in \zo^n$, we have
\[
  \prod_{i=1}^n g_i^{x_i} = \bigl(f_1(\overline{g}),f_{2}(\overline{g}),\ldots,f_k(\overline{g})\bigr)(x_1, \ldots, x_n) .
\]
\end{lemma}

This lemma is essentially in the previous work \cite{DBLP:journals/eccc/ECCC-TR01-040}.
However the statement there is for nilpotent groups and cannot be
immediately used. Also, the proof relies on previous work and is somewhat
indirect. So we give a direct proof of the result we need (i.e., \Cref{lem:p-group-word-polynomial}).

Before the proof we illustrate it via an example.
\begin{example}
Let $G := \mathbb{Z}_{2} \wr \mathbb{Z}_{2}$ from the introduction.
Consider a product $\prod_i (a_i,b_i;z_i)$.  Via a polynomial map we can rewrite this product into a normal form where all the $z_i$ are in one element only:
\[
\Bigl( \prod_i (a'_i,b'_i;0) \Bigr) (0,0;z').
\]
Computing this product is then immediate, via a linear map.  The key observation is that $a'_i = a_i$ if the sum of the $z_j$ with $j<i$ is even, and $a'_i = b_i$ otherwise, and that this computation is a quadratic polynomial (in the input bits $a_i,b_i,z_i$).  
\end{example}

\begin{proof}[Proof of \Cref{lem:p-group-word-polynomial}]
We proceed by induction on $k$.
If $k=1$, then $G$ is cyclic.
We can take a generator $a \in G$ and define the 1-1 mapping $G \ni a^z \leftrightarrow z \in \F_p$.
So $\prod_{i=1}^n g_i^{x_i} = \prod_{i=1}^n a^{z_i x_i}$ can be written as the degree-1 polynomial $\sum_{i=1}^n z_i x_i$.

Otherwise, $G$ has a normal subgroup $H$ of order
$p^{k-1}$~\cite[Chapter~6, Theorem~1.(3)]{dummitfoote}. The corresponding
quotient group $Q = G/H$ has order $p$ and is therefore cyclic.
So we can write $g_i \in G$ as
\[
g_{i}=a^{e_{i}}h_{i}
\]
where $h_{i}\in H$ and $a$ is a generator of $Q$.
Applying the induction hypothesis on $H$, we can identify each element $g_i = a^{e_i} h_i$ with a $k$-tuple $(e_i, e_i') \in \F_p^k$, where $e_i' \in \F_p^{k-1}$ corresponds to $h_i \in H$.

Now we apply the conjugation trick as in \cite{Barrington89}, and use induction. That is, let $b_{i}:=a^{\sum_{j\le i} e_j x_j}$ and write
\[
  \prod_{i=1}^n g_i^{x_i}
  = \Bigl(\prod_{i=1}^n \left(b_i h_i^{x_i} b_i^{-1}\right)\Bigr)b_n.
\]
Note that $b_{i}$ and $b_i^{-1}$ can be computed by some degree-1 polynomials over $\F_p$, and $h_i^{x_i}$ can be (trivially) computed by a degree-$c_H$ polynomial over $\F_p$.

Therefore, each term $b_i h_i^{x_i} b_i^{-1}$ can be computed by some degree-$c_G$ polynomial map $f^{H,i} = (f_1, \ldots, f_k)$ over $\mathbb{F}_{p}$.
Moreover, these terms lie in $H$ because $H$ is a normal subgroup of $G$. Hence we have reduced to a product over
$H$, which by induction hypothesis, can be computed by some degree-$c_H$ polynomial, and the result follows.
\end{proof}

Given \Cref{lem:p-group-word-polynomial}, it suffices to construct a \emph{bit}-generator that fools low-degree polynomials over $\F_p$.

\paragraph{The case $p=2$.}
For this case, we can simply combine \Cref{lem:p-group-word-polynomial}
with known generators for polynomials over $\mathbb{F}_{2}$ \cite{BoV-gen,Lov09,Viola-d}.  In fact, we obtain results for non-read-once programs
as well, and of any length. (Indeed, such polynomials are equivalent
to low-degree polynomials over $\mathbb{F}_{2}$.)

\paragraph{The case $p>2$.}
Here we need additional ideas because \emph{bit}-generators that fool polynomials over
$\mathbb{F}_{q}$ with $q\ne2$ are not known. However, the works
\cite{BoV-gen,Lov09,Viola-d} do give generators that output \emph{field elements} that fool such polynomials.

\begin{lemma}[\cite{Viola-d}]
\label{lem:viola-d}There are distributions $Y$ over $\mathbb{F}_{p}^{n}$
that can be explicitly sampled from a uniform seed of $c_{p}(2^{d} \log(1/\e) +\log n)$
bits such that for any degree-$d$ polynomial $f$ in $n$ variables over
$\mathbb{F}_{p}$, we have $\Delta(f(Y),f(U))\le\e$.
\end{lemma}

However, we need distributions over $\zo^n$. This distinction is critical
and arises in a number of previous works. Currently, for domain $\zo^n$
only weaker results with seed length $\ge\log^{2}n$ are known \cite{LovettMS10}.

Still, as pointed out in \cite{LovettRTV09,MekaZ09}, \Cref{lem:viola-d}
implies results over the domain $\F_2^n$ for \emph{biased bits:}

\begin{definition}
We denote by $N_p$ a vector of $n$ i.i.d.~bits
coming up $1$ with probability $1/p$.
\end{definition}

\begin{corollary}[\cite{LovettRTV09,MekaZ09}] \label{cor:N_p-fool-polynomial}
  There are distributions $X$ over $\zo^{n}$ that can be explicitly sampled
  from a uniform seed of $c_{p}(2^{d}\log(1/\eps)+\log n)$ bits such that for
any degree-$d$ polynomial $f$ in $n$ variables over $\mathbb{F}_{p}$
we have $\Delta(f(X),f(N_p))\le\e$.
\end{corollary}

\begin{proof}
Let $Y = (Y_1, Y_2, \ldots, Y_n)$ be the distribution from \Cref{lem:viola-d},
for degree $d(p-1)$. Define $X:=(Y_{1}^{p-1},Y_{2}^{p-1},\ldots,Y_{n}^{p-1})$.
Note $X$ is over $\zo^{n}$. Also, if $U$ is uniform in $\mathbb{F}_{p}$
then $U^{p-1}=N_{p}$. The result follows.
\end{proof}
We will show how to use biased bits. For this we use that the program
is read-once.
\begin{lemma}
\label{lem:FK}Let $X$ fool degree-1 polynomials over $\mathbb{F}_{2}$
with error $\e^{c_{G,p}}$. Then $X+N_{p}$ fools programs of length $n$ over $G$
with error $\e$.
\end{lemma}

\begin{proof}
This follows from Lemma 7.2 in \cite{ForbesKelley-2018} combined
with the Fourier bound in \cite{ReingoldSV13,LPV22}.
The proof in \cite{ForbesKelley-2018}
is for the fixed noise parameter $p=4$, but the generalization
to any $p$ is immediate (replace $1/2$ with $1-2/p$ in the last
two lines of the proof).
\end{proof}
We now have all the ingredients.
\begin{proof}[Proof of \Cref{thm:fooling-p-groups}.]
 Use \Cref{lem:FK}. Averaging over $X$, it suffices to derandomize
$N_{p}$. By \Cref{lem:p-group-word-polynomial} it suffices to do this
for low-degree polynomials. This follows from \Cref{cor:N_p-fool-polynomial}.
\end{proof}

\section{Representation theory and matrix analysis}

In this section, we present the fragment of representation theory and matrix analysis that we need. The books by Serre~\cite{Serre77}, Diaconis~\cite{MR964069}, and Terras~\cite{MR1695775} are good references for representation theory and
non-commutative Fourier analysis. The Barbados notes~\cite{Wigderson-Barbados10}, \cite[Section~13]{MR3584096}, \cite{GowersV-group-mix}, or \cite{DLV-group-bu} provide briefer introductions.

\paragraph{Matrices.}
Let $M$ be a square complex matrix. We denote by $\tr(M)$
the trace of $M$, by $\overline{M}$ the conjugate of $M$, by $M^{T}$
the transpose of $M$, and by $M^{*}$ the conjugate transpose $\overline{M^{T}}$
(aka adjoint, Hermitian conjugate, etc.). The matrix $M$ is \emph{unitary}
if the rows and the columns are orthonormal; equivalently $M^{-1}=M^{*}$.

The \emph{Frobenius norm}, (a.k.a.~Schatten 2-norm, Hilbert--Schmidt operator) of a square matrix $M$, denoted $\norm{M}_\Fr$, is  $\sum_{i,j}\abs{M_{i,j}}^2 = \tr(M M^*)$.

The \emph{operator norm} of a matrix $M$, denoted $\norm{M}_\op$, is the square root of the largest eigenvalue of the matrix $MM^\ast$.
In particular, if $M$ is a normal matrix, i.e. $MM^\ast = M^\ast M$, then $\norm{M}_\op$ equals its largest eigenvalue in magnitude.

\begin{fact}  $\norm{AB}_\op \le \norm{A}_\op \norm{B}_\op$.
\end{fact}

\begin{fact} \label{fact:operator-norm-bounds-frobenius}
For a $d \times d$ matrix $M$ with eigenvalues $\lambda_1, \ldots, \lambda_d$, we have $\norm{M}_\Fr^2 = \sum_{i=1}^d \abs{\lambda_i}^2 \le d \norm{M}_\op^2$.
\end{fact}

\paragraph{Representation theory.}

Let $G$ be a group. A \emph{representation }$\rho$ of $G$ with
dimension $d$ maps elements of $G$ to $d\times d$ unitary, complex
matrices so that $\rho(xy)=\rho(x)\rho(y)$. Thus, $\rho$ is a homomorphism
from $G$ to the group of linear transformations of the vector space
$\mathbb{C}^{d}$. We denote by $d_{\rho}$ the dimension of $\rho$.

If there is a non-trivial subspace $W$ of $\mathbb{C}^{d}$ that
is invariant under $\rho$, that is, $\rho(x)W\subseteq W$ for every
$x\in G$, then $\rho$ is \emph{reducible}; otherwise it is \emph{irreducible.}
Irreducible representations are abbreviated \emph{irreps }and play
a critical role in Fourier analysis. We denote by $\widehat{G}$ a
complete set of inequivalent irreducible representations of $G$.

Let $\hG$ be the set of irreducible representations of $G$ (i.e. the dual group of $G$).
We have
\begin{equation}
\sum_{\rho \in \hG} d^2_\rho = |G|.
    \label{eq:sum-d-square}
\end{equation}

  For a random variable $Z$ we also use $Z$ to denote its probability mass function.

  For an irrep $\rho \in \hG$, the $\rho$-th Fourier coefficient of $Z$ is
  \[
    \hZ(\rho) := \sum_{g \in G} Z(g) \ol{\rho(g)}  = \E\Bigl[\ol{\rho(Z)}\Bigr] .
  \]
  The Fourier expansion of $Z \colon G \to \R$ is
  \[
    Z(g) = \frac{1}{\abs{G}} \sum_{\rho \in \hG} d_\rho \tr\bigl(\hZ(\rho) \rho(g) \bigr) .
  \]
  Parseval's identity gives
  \[
    \sum_{g \in G} Z(g)^2
    = \frac{1}{\abs{G}} \sum_{\rho \in \hG} d_\rho \norm{\hZ(\rho)}_\Fr^2 .
  \]

\begin{claim} \label{claim:closeness}
Suppose $X$ and $Y$ are two random variables over $G$ such that for every irreducible representation $\rho$ of $G$, we have $\norm{\E[\rho(X)] - \E[\rho(Y)]}_\op \le \eps$.
 Then $X$ and $Y$ are $(\sqrt{\abs{G}} \cdot \eps)$-close in statistical distance.
\end{claim}
\begin{proof}
  \begin{align*}
  \frac{1}{2} \sum_{g \in G} \abs[\big]{X(g) - Y(g)}
    &\le \frac{\sqrt{\abs{G}}}{2} \Bigl( \sum_{g \in G} \bigl ( X(g) - Y(g) \bigr)^2 \Bigr)^{1/2} & \text{(Cauchy--Schwarz)}\\
    &= \frac{\sqrt{\abs{G}}}{2} \biggl( \frac{1}{\abs{G}} \sum_{\rho \in \hG} d_\rho \norm[\big]{ \hX(\rho) - \hY(\rho) }_\Fr^2 \biggr)^{1/2} & \text{(Parseval)} \\
    &= \frac{1}{2} \Bigl( \sum_{\rho \in \hG} d_\rho \cdot (d_\rho \eps^2) \Bigr)^{1/2} & \text{(\Cref{fact:operator-norm-bounds-frobenius})}\\
    &= \frac{\eps}{2} \cdot \Bigl(\sum_{\rho \in \hG} d_\rho^2\Bigr)^{1/2}
    = \sqrt{\abs{G}} \cdot \eps/2 . &\text{(\Cref{eq:sum-d-square})} & \qedhere
  \end{align*}
\end{proof}

\section{Proof of \Cref{thm:fooling-mixing-groups}}

Again, besides the parameter improvement, our main point here is to illustrate how we use representation theory to obtain pseudorandom generators.  These ideas will then be generalized to the more general and complicated setting of block products in the next section.

Let $\rho$ be an irreducible representation of a mixing group (\Cref{def:mixing}).
By definition of mixing, if $\rho$ is a non-identity matrix then it does not have $1$ as its eigenvalues.
A main observation is that if there are many non-identity matrices $\rho(g_i)$ in the program, then the bias $\norm{\E[\prod_{i=1}^n \rho(g_i)^{U_i}]}_\op$ is small.  This is proved in the next two claims.

\begin{claim} \label{claim:avg-norm}
  Let $M$ be a unitary matrix with eigenvalues $e^{i\theta_j}$ for some $\theta_j \in [-\pi,\pi]$ on the unit circle.
Suppose $\abs{\theta_j} \ge \theta$ for every $j$. Then $\norm{(I + M)/2}_\op \le 1 - \theta^2/8$.
\end{claim}

\begin{proof}
  As $M$ is unitary, we can write $M = Q^\ast D Q$, where $D$ is a diagonal matrix with $M$'s eigenvalues on its diagonal and $Q$ is unitary.
  The eigenvalues of $(I+M)/2 = Q^\ast (\frac{I+D}{2})Q$ are $\frac{1 + e^{i\theta_j}}{2} = e^{i\theta_j/2} \cdot \frac{e^{-i\theta_j/2} + e^{i\theta_j/2}}{2} = e^{i\theta_j/2} \cos(\theta_j/2)$, which have magnitudes at most $\abs{\cos(\theta_j/2)} \le 1 - \theta^2/8$.
\end{proof}

\begin{claim} \label{claim:irreps-mixing}
  Let $G$ be a mixing group.
  Let $\rho$ be an irreducible representation of $G$ of dimension $d_\rho$.
  Let $f_\rho(x) = \prod_{i=1}^n \rho(g_i)^{x_i}$ be the representation of a group program.
  Suppose $\rho(g_i) \ne I_{d_\rho}$ for $t \ge c_G \log(1/\eps)$ many $i$'s.
  Then $\norm{\E[f_\rho(U)]}_\op \le \eps$.
\end{claim}

\begin{proof}
  Let $T$ be the $t$ coordinates $j$ where $\rho(g_j) \ne I_{d_\rho}$.
  For every fixing of the other coordinates, we can write $f_\rho(U)$ as
  \[
    B \prod_{j \in T} \rho(g_j)^{U_j} B_j
  \]
  for some unitary matrices $B$ and $B_j$'s.
  So
  \[
    \norm{f_\rho(U)}_\op
    \le \norm{B}_\op \prod_{j \in T} \norm{ \E[\rho(g_j)^{U_j}]}_\op \norm{B_j}_\op
    \le (1 - c_G)^t \le \eps . \qedhere
  \]
\end{proof}

We now proceed with the proof of the main result.  The proof extends to handle the spill, but for simplicity we do not discuss it here. We fool each irreducible representation of $G$ separately and then appeal to \Cref{claim:closeness}.
Fix a representation $\rho$ and consider the product
\[
  f_\rho(x) := \prod_{j=1}^\ell \rho(g_j)^{x_j} .
\]
Let $t$ be the number of non-identity elements $\rho(g_j) : j \in [n]$ and $S$ be their coordinates.

Let us sketch the construction.
First, XORing with an almost $2c_G \log(1/\eps)$-wise uniform distribution takes care of the case $t \le c_G \log(1/\eps)$, so we may assume that $t$ is larger.
In this case, by \Cref{claim:irreps-mixing}, we have that the bias $\norm{\E[f_\rho(U)]}_\op$ is small under the uniform distribution.
Our goal is to set $ct$ bits in $S$ to uniform and apply \Cref{claim:irreps-mixing} again.

Let $\ell := c_G \log(1/\eps)$.
Let $M$ be a $(\log n) \times 10\ell$ matrix filled with uniform bits.

We will make $\log n$ guesses of $t$.
For each guess $v=2^i \cdot \ell : i\in \{0, \ldots,\log n-1\}$ of $t$, we select a subset of size $\ell$ of the input positions using a hash function $h_i$, and then hash these $\ell$ positions to row $i$ of $M$ using another hash function $h$, and assign input bits correspondingly.
The final generator is obtained by trying all guesses, using the same seed for each guess $h_i$, and XORing
together the bits.

\medskip

In more detail, for each $i \in \{0, \ldots, \log n-1\}$, let $h_i\colon [n] \to \zo$ be a $10\ell$-wise independent hash family with $\Pr_{h_i}[h_i(j) = 1] = 2^{-i}$ for each $j \in [n]$.
Let $h\colon [n] \to [10\ell]$ be another $5\ell$-wise uniform hash family.
The output of our generator is 
\[
  D := D^{(0)} \oplus \cdots \oplus D^{(\log n-1)},
\]
where the $j$-th bit of $D^{(i)}$ is $$h_i(j) \cdot M_{i,h(j)}.$$
We use the same seed to sample $h_0, \dots, h_{\log n-1}$, which costs at most $O_G(\log n \log(1/\eps))$ bits~\cite[Corollary~3.34]{Vadhan12}.
Sampling $h$ uses another $O_G(\log(n/\eps) \log(1/\eps))$ bits.
This uses a total of $O_G(\log(n/\eps) \log(1/\eps))$ bits.

\medskip

We now show that $\norm{\E[f_\rho(D)]}_\op \le O(\eps)$.
Suppose $t \in [2^i \ell, 2^{i+1} \ell]$.
Recall that $S$ are the coordinates corresponding to the non-identity matrices in the product.
Let $J := h_i^{-1}\{1\} \cap S$.
As $\Pr[h_i(j) = 1]=2^{-i}$, we have $\E[\abs{J}] \in [\ell, 2\ell]$.
Applying tail bounds for bounded independence (see \Cref{lem:k-wise-tail}), we have $\abs{J} \in [\ell/2, 3\ell]$ except with probability $\eps$.
Conditioned on this event, as $\abs{J} \le 3\ell$ and $h$ is $5\ell$-wise uniform, we can think of $h$ as a random function from $J$ to $[10\ell]$.
Hence, for each $j \in [10\ell]$, we have
\[
  \Pr\bigl[ \abs{J \cap h^{-1}(j)} = 1\bigr]
  = \abs{J} \cdot 1/(10\ell) \cdot (1 - 1/(10\ell))^{\abs{J}-1}
  \ge (\ell/2) \cdot (1/10\ell) \cdot (1/2)
  \ge 1/40 .
\]
By a Chernoff bound, we have that except with probability at most $\eps$, the number of $j$ such that $\abs{J \cap h^{-1}(j)} = 1$ is at least $\ell/10$.

Let $T$ be these coordinates.
Fixing all the bits in $M$ except the ones in row $i$ that are fed into $T$, we can write the conditional expectation of $f_\rho(G)$ over the bits in $T$ as
\[
  B \prod_{j \in T} \E_{x_j}[A_j^{x_j}] B_j,
\]
for some unitary matrices $B$, $A_j$'s and $B_j$'s, and in particular, $A_j$ has its eigenvalues bounded away from 1 on the complex unit circle.
Therefore, by \Cref{claim:avg-norm},
\[
  \norm[\Big]{ B \prod_{j \in T} \E_{x_j}\bigl[A_j^{x_j}\bigr] B_j }_\op
  \le \prod_{j \in T} \norm[\bigg]{\frac{(I+A_j)}{2}}_\op
  \le \eps . \qedhere
\]

\section{Proof of \Cref{thm:reduction}}
In this section we prove \Cref{thm:reduction}.  This type of reductions goes back to the work of \cite{GopalanMRTV12} on read-once CNFs (itself building on \cite{AjtaiW89}), and have been refined in several subsequent works.  The work \cite{LeeV-rop} extended the techniques to read-once polynomials.  It exploited the observation that when the number of monomials is significantly larger than its degree, the bias of the polynomial is small, and therefore the bias of the restricted function remains small.
Building on this observation, \cite{MRT19} showed that one can aggressively restrict most of the coordinates, while keeping the bias of the restricted function small.
In addition, a typical restricted product is a low-degree polynomial (plus a spill), for which we have optimal generators~\cite{BoV-gen,Lov09,Viola-d}.

However, \cite{MRT19} reduces to non-linear polynomials (degree 16).  As discussed earlier, bit-generators with good seed lengths are only known over $\Z_2$.  We give a refined reduction that reduces to polynomials of degree one, for which we have generators over $\Z_m$ for any $m$ \cite{LovettRTV09,MekaZ09,GKM18}.

At the same time, we show that the reduction can be carried over any mixing group, by working with representations of the group.

\begin{definition}
  Let $\calU_\theta(d)$ be the set of $d \times d$ unitary matrices with eigenvalues $e^{2\pi i\theta_j}$ where $\abs{\theta_j} \ge \theta$.
\end{definition}

\begin{definition}
A group $G$ is \emph{$\theta$-mixing} if it has a complete set of unitary
irreducible representations where each non-identity matrix lies in $\calU_\theta(d)$ for some $d$.
\end{definition}

The following theorem will serve as the basis of our iterative construction of the PRG.

\begin{theorem} \label{thm:one-iter}
  Let $w \ge \log\log(1/\eps) + \log m$.
  Suppose there is a PRG $P$ with seed length $s$ that $\eps$-fools $(m^5 2^{30\width},2\width, 2\log(1/\eps))$-products over $G$.
  Let $P_1$ be a PRG with seed length $s_1$ that $\eps$-fools $(\terms,1,3\log(1/\eps))$-products over a group $G$ of order $m$ that is $(1/m)$-mixing.
  Then there is a PRG that $O(\eps)$-fools $(m^5 2^{45\width}, 3\width, 2\log(1/\eps))$-products over $G$ with seed length
  \[
    s + \Bigl( s_1 + O_m((\log(1/\eps) + w) \log w + \log\log n) \Bigr) .
  \]
\end{theorem}

We first show how to apply \Cref{thm:one-iter} iteratively to obtain \Cref{thm:reduction}.

\begin{proof} [Proof of \Cref{thm:reduction}]
  We iterate \Cref{thm:one-iter} repeatedly for some $t$ times to reduce the problem to fooling an $O(\log(m/\eps))$-junta which can be done using an almost bounded uniform distribution.

  Given an $(\terms,\width,\log(1/\eps))$-product $f$, let $\width' = \max\{\width, \log \terms, \log m\}$ so that we can view $f$ as an $(m^5 2^{45\width'}, 3\width', 2\log(1/\eps))$-product.
  We first apply \Cref{thm:one-iter} for $t_1 = O(\log \width')$ times until we have a 
  \[
    \bigl( (m \cdot \log(1/\eps))^C, \log\log(1/\eps) + \log m, 2\log(1/\eps)\bigr)\text{-product},
  \]
  for some constant $C$.

  Let $b := \frac{\log(1/\eps) + \log m}{\log\log(1/\eps) + \log m}$.
  We will apply the following repeatedly for some $r = O_m(1)$ steps.
  We divide the $f_i: i \ge 1$ into groups of $b$ functions and view the product of functions in each group as a single function, this way we can think of the above product as a 
  \[
    \Bigl(\tfrac{ (m \cdot \log(1/\eps) )^C}{b}, \log(1/\eps) + \log m, 2\log(1/\eps) \Bigr)\text{-product.}
  \]
  So we can continue applying \Cref{thm:one-iter} for $t_2 = O(\log(\log(1/\eps) + \log m)) \le O_m(\log\log(1/\eps))$ times and the restricted function becomes a
  \[
    \Bigl(\tfrac{ (m \cdot \log(1/\eps) )^C}{b}, \log\log(1/\eps) + \log m, 2\log(1/\eps) \Bigr)\text{-product.}
  \]
  Repeating this process for 
  \begin{align*}
    r
    = \log_b\Bigl( \bigl(m \cdot \log(1/\eps) \bigr)^C \Bigr)
    \le \frac{2C \bigl( \log m + \log\log(1/\eps) \bigr)}{ \log\log(1/\eps) }
    = O_m(1)
  \end{align*}
  times, we are left with a
  \[
    \bigl( O(1), \log\log(1/\eps) + \log m, 2\log(1/\eps) \bigr)\text{-product}
  \]
  which can be fooled by an $\eps$-almost $O(\log(m/\eps))$-wise uniform distribution that can be sampled using $s' = O(\log(m/\eps) + \log\log n)$ bits~\cite{NaN93,AGHP92}.
  Therefore, in total we apply \Cref{thm:one-iter} for
  \begin{align*}
    t := t_1 + r \cdot t_2
    \le  O(\log \width') + O_m( \log\log(1/\eps) )
    = O_m\bigl(\log(w + \log(\ell/\eps))\bigr)
  \end{align*}
  times, each with a seed of
  \[
    s = s_1 + O_m((\log(\ell/\eps) + w) \log w + \log\log n) .
  \]
  bits.
  Hence in total it uses 
  \begin{align*}
    s \cdot t + s'
    &\le O_m\bigl( s_1 +  \log(\ell/\eps) + \width  \bigr) \cdot \polylog\bigl(\width, \log \terms, \log n, \log(1/\eps)\bigr) 
  \end{align*}
  bits.
\end{proof}

\subsection{Analysis of one iteration: Proof of \Cref{thm:one-iter}}

We now prove \Cref{thm:one-iter}.
Given an $(m^5 \cdot 2^{45\width}, 3\width, 2\log(1/\eps))$-product $f = \prod_{i=0}^\terms f_i$ over $G$ of order $m$ that is $(1/m)$-mixing, let $\terms$ be the number of non-constant $f_i$.
We say $f$ is a \emph{long product} if $\terms \ge m^5 \cdot 2^{30\width}$, otherwise $f$ is \emph{short}.
At a high-level, we apply \Cref{thm:width-reduction} to $P_1$ to obtain a PRG that fools long products in \emph{one shot}, and use \Cref{lem:sbpnfp,lem:reduce-short-prod} below to reduce fooling a short product to fooling a product of smaller width $\width$.

\begin{lemma} \label{lem:sbpnfp}
  Let $w \ge \log m$ and $C$ be a sufficiently large constant.
  Define
	\begin{align*}
    k &:= C (w + \log(\ell/\eps)) \\
    \delta &:= (m \cdot w)^{-k} \\
		p &:= 2^{-C} .
	\end{align*}
  There exist two $\delta$-almost $k$-wise independent distributions $D$ and $T$ with $\E[D_i]=1/2$ and $\E[T_i]=p$ for every $i\in [n]$, such that for every $(\ell, w, 0)$-product $f$ over $G$ of order $m$, we have $\abs{\E_{D,T} [ f(D + T \wedge U) ] -\E[f(U)] } \le \eps$.

  Moreover, $D$ and $T$ can be efficiently sampled with a seed of length $O_m((\log(\ell/\eps) + w) \log w + \log\log n)$.
\end{lemma}

\Cref{lem:sbpnfp} follows from the following lemma, which can be obtained from applying a variant of a result of Forbes and Kelley~\cite{ForbesKelley-2018} to the Fourier bounds on functions computable by block products over groups, which was established in \cite{LPV22}.
(Block products are called generalized group products in \cite{LPV22}.)

\begin{lemma}[\cite{ForbesKelley-2018,LPV22}]
  Let $f\colon \zo^n \to \zo$ be computable by an $(\ell,w,0)$-block product over a group $G$.
  Let $D$ and $T$ be two independent $\delta$-almost $2(k+w)$-wise independent distributions on $\zo^n$ with $\E[D_i] = 1/2$ and $\E[T_i] = p$,
  and $U$ be the uniform distribution on $\zo^n$.
  Then
  \[
    \abs[\big]{\E[f(D + T \wedge U)] - \E[f(U)] }
    \le \ell \cdot \bigl( \sqrt{\delta} \cdot (w \cdot \abs{G})^{k+w} + (1-2p)^{k/2} + \sqrt{\gamma} \bigr) .
  \]
\end{lemma}

We remark that for every constant $p$, one can show that $n^{-\omega(1)}$-bias plus noise $N_p$ is necessary to fool programs over groups of order $\poly(n)$ with any subconstant error $\eps$.
This follows from \cite{DILV24-II}, where it shows that there exists such a distribution which puts $2\eps$ more probability mass on strings whose Hamming weight is greater than $n/2 + O_p(\sqrt{kn})$ than the uniform distribution.

\begin{lemma}[Width reduction for short products] \label{lem:reduce-short-prod}
  Let $G$ be any group of order $m$.
  Let $D$ and $T$ be the two distributions defined in \Cref{lem:sbpnfp}.
  Let $w \ge \log m$.
  Let $f$ be an $(\terms,3\width,2\log(1/\eps))$-product over $G$, where $\terms \le m^5 \cdot 2^{30\width}$.
  Then with probability at least $1 - \eps$ over $D$ and $T$, the restricted function $f_{D,T}$ is an $(\terms,2\width,2\log(1/\eps))$-product over $G$.
\end{lemma}

We prove \Cref{lem:reduce-short-prod} in \Cref{sec:short-prod}.

\begin{theorem} \label{thm:width-reduction}
  Let $w \ge \log\log(1/\eps)+ 2\log(1/\theta) $.
  Suppose there is a PRG $P_1$ with seed length $s_1$ that $\eps$-fools $(\terms,1,3\log(1/\eps))$-product $f := \prod_{i=0}^\terms f_i$ over a matrix group $\calM$ supported on $\calU_\theta(d)$, where $\terms \ge 2^{3\width} \theta^{-2}$ and each $f_i$ is non-constant.
  Then there is a PRG that fools $(\terms,3\width,2\log(1/\eps))$-products $f = \prod_{i=0}^\ell f_i$ over $\calM$, where $\terms \in [2^{30\width} \theta^{-5}, 2^{45\width} \theta^{-5}]$ and every $f_i$ is non-constant, with seed length $s = s_1 + O_\theta\bigl( \log(1/\eps) + \width + \log\log n \bigr)$.
\end{theorem}

\begin{corollary}
  \Cref{thm:width-reduction} applies to products over any $\theta$-mixing group $G$ with $\eps$ replaced by $\eps/\sqrt{\abs{G}}$.
\end{corollary}
\begin{proof}
  By definition, all its irreps $\rho$ belong to $\calU_\theta(d_\rho)$.
  It follows from \Cref{claim:closeness} that it suffices to fool all its irreps.
\end{proof}

We prove \Cref{thm:width-reduction} in \Cref{sec:width-reduction}.
We now show how \Cref{thm:one-iter} follows from \Cref{lem:sbpnfp,lem:reduce-short-prod} and \Cref{thm:width-reduction}.

\begin{proof}[Proof of \Cref{thm:one-iter}]
  Let $P_1$ be the PRG that $\eps$-fools $(\ell,1,3\log(1/\eps))$-products with seed length $s_1$.
  Applying \Cref{thm:width-reduction} with $P_1$, we obtain a PRG $P_{\LONG}$ that $\eps$-fools every $(\ell,3w,2\log(1/\eps))$-product $f = \prod_{i=0}^\ell f_i$, where $\ell \in [m^5 2^{30w}, m^5 2^{45w}]$ and every $f_i$ is non-constant, with seed length $s_\LONG = s_1 + O_m(\log(1/\eps) + w + \log\log n)$.
	We now sample the distributions $D,T$ in \Cref{lem:sbpnfp}, and output
	\[
    (D + T \wedge P(U)) + P_\LONG .
	\]
  Using \Cref{lem:sbpnfp} and $\ell \le 2^{O_m(w)}$, sampling $D$ and $T$ uses $s_\SHORT = O_m((\log(1/\eps) + w)\log w + \log\log n)$ bits.
  So altogether this takes $s + s_\LONG + s_\SHORT = s + s_1 + O_m((\log(1/\eps) + w) \log w + \log\log n)$ bits.

\smallskip

  Let $f$ be an $(m^5 2^{45\width}, 3\width, 2\log(1/\eps))$-product with $\terms$ many non-constant $f_i$'s.
	If $\terms \ge m^5 2^{30\width}$, then $P_\LONG$ $\eps$-fools it.
  Otherwise, $\terms \le m^5 2^{30\width}$ and so $f$ is an $(m^5 2^{30\width}, 3\width, 2\log(1/\eps))$-product.
  So by \Cref{lem:reduce-short-prod}, with probability at least $1-\eps$ over the choices of $D$ and $T$, the function $f_{D,T}$ is an $(m^5 2^{30\width}, 2\width, 2\log(1/\eps))$-product, and therefore can be $\eps$-fooled using the generator $P$ given by the assumption.
  The total error is $O(\eps)$.
\end{proof}

\section{Width reduction for long products: Proof of \Cref{thm:width-reduction}} \label{sec:width-reduction}
In this section, we prove \Cref{thm:width-reduction}.
Let $f = \prod_{i=0}^\ell f_i$ be an $(\terms,3\width,2\log(1/\eps))$-product over a matrix group $\calM \subseteq \calU_\theta(d)$, where $\terms \in [2^{30\width}\theta^{-5}, 2^{45\width}\theta^{-5}]$ and each $f_i$ is non-constant.
Note that when a product $f$ has this many non-constant functions, the ``bias'' $\norm{\E[f(U)]}_\op$ of $f$ is \emph{doubly exponentially small} in $\width$, i.e. at most $\exp(-2^{2\width})$ (see \Cref{claim:long-prod-small-bias}, which is at most $\eps$ whenever $\width \ge \log\log(1/\eps)$).
Following \cite{MRT19}, we will pseudorandomly restrict most of the coordinates of $f$ and show that the bias of a typical restricted product remains bounded by $\eps$.
More importantly, we will show that this restricted product has \emph{width 1} (with a small spill).
Therefore, it suffices to construct a PRG for width-1 products (with a small spill).

We remark that previous works showed that a typical restricted product has degree at most 16, as opposed to 1.
This difference is already crucial in fooling products over $\Z_m$ for composites $m$ with good seed lengths, as we do not have (bit)-PRGs even for degree-2 polynomials over $\Z_6$.

\subsection{The reduction}

We will use the following standard construction of $\delta$-almost $k$-wise independent distributions with marginals $p$.
\begin{claim} \label{claim:almost-indept-w-marginals}
  There exists an explicit $\delta$-almost $k$-wise independent distribution $T$ on $\zo^n$ with $\E[T_i] = 2^{-b}$ for every $i \in [n]$ which can be sampled using $O(b + k + \log(1/\delta) + \log\log n)$ bits.
\end{claim}
\begin{proof}
  We sample an $(\delta,kb)$-biased distribution $D$ on $\zo^{nb}$ and $b$ uniform bits $U_b$.
  By standard construction~\cite{NaN93,AGHP92}, $D$ can be sampled using $O(b + \log(k/\eps) + \log\log n)$ bits.
  Write $D = (D_1, \ldots, D_n)$ where each $D_i \in \zo^b$.
  We output $T \in \zo^n$, where $T_i = \AND_b(D \oplus U_b)$, where $\AND_b$ is the $\AND$ function on $b$ bits.
  We have $\E[T_i] = 2^{-b}$ because $U_b$ is uniform.
  By \cite[Claim 3.7]{MRT19}, $T$ is $(\eps \cdot 2^k)$-almost $k$-wise uniform.
  Setting $\eps = 2^{-k} \cdot \delta$ proves the claim.
\end{proof}

Let $C$ be a sufficiently large constant.
Let
\begin{align}
  k &= C (\log(1/\eps) + \width) \nonumber \\
  \delta &= \theta^k  \label{eq:long-prod-params}\\
  p &= 2^{-23\width} \theta^3 . \nonumber 
\end{align}
Let $D$ and $T$ be two $\delta$-almost $k$-wise independent distributions, with $\E[D_i]=1/2$ and $\E[T_i] = p$ for every $i \in [n]$, and let $P_1$ be the PRG given by the theorem.
The generator is 
\[
  P := D + T \wedge P_1 .
\]
By \Cref{claim:almost-indept-w-marginals}, this uses $s_1 + O_\theta(\log(1/\eps) + \width + \log\log n)$ bits.

\subsection{Analysis}

We first state a claim showing that if the number of non-constant $f_i$'s $\ell$ in a block product $f$ is much greater than its width $\width$, then the bias $\norm{\E[f(U)]}_\op$ is small.
We defer its proof to \Cref{sec:long-prod-small-bias}.

\begin{claim} \label{claim:long-prod-small-bias}
  For integers $\width$ and $q$, let $f = \prod_{i=0}^\ell f_i$ be an $(\terms,\width,q)$-product over some matrix group $\calM$ supported on $\calU_\theta(d)$ for some $\terms \ge 2^{2\width+2}\theta^2 \log(1/\eps)$, where each $f_i$ is non-constant.
  Then $\norm{\E[f(U)]}_\op \le \eps$.
\end{claim}

Recall that $\terms \in [2^{30\width} \theta^{-5}, 2^{45 \width} \theta^{-5}]$.
Given $D,T$, let $f_{D,T}\colon \zo^T \to \calM$ be the restricted product
\[
  f_{D,T}(x) := f(D + T \wedge x) = \prod_{i=0}^\terms f_i(D + T \wedge x) .
\]
We use $f_{D,T,i}(x)$ to denote $f_i(D + T \wedge x)$.

The following lemma shows that with high probability over $D$ and $T$, the function $f_{D,T}$ is an $(\terms,1,3\log(1/\eps))$-product, that is, a group program with a small spill.
Note that this lemma is true for products over any group.

\begin{lemma} \label{lem:1bit-product}
  Let $D$ and $T$ be two distributions on $\zo^n$ defined in \cref{eq:long-prod-params}.
  Let $w \ge \log\log(1/\eps) + 2\log(1/\theta)$ and $f = \prod_i f_i$ be an $(\terms,3\width,2\log(1/\eps))$-product, where $\terms \in [2^{30\width}\theta^{-5}, 2^{45 \width}\theta^{-5}]$ and each $f_i$ is non-constant.
  Then with probability $1 - \eps$ over $D$ and $T$, the function $f_{D,T}$ is an $(\terms,1,3\log(1/\eps))$-product, where $\terms \ge 2^{3\width} \theta^{-2}$ and each $f_{D,T,i}$ is non-constant.
\end{lemma}

\Cref{thm:width-reduction} follows from \Cref{claim:gkm-prg,claim:long-prod-small-bias,lem:1bit-product}.

\begin{proof}[Proof of \Cref{thm:width-reduction}]
  As $\terms \ge 2^{2\width+2} \theta^{-2} \log(1/\eps)$, by \Cref{claim:long-prod-small-bias}, we have $\norm{\E[f(U)]}_\op \le \eps$.
  By \Cref{lem:1bit-product}, with probability $1-\eps$ over $D$ and $T$, the restricted function $f_{D,T} = \prod_i f_{D,T,i}$ is an $(\terms,1,3\log(1/\eps))$-product, where $\terms \ge 2^{3\width} \theta^{-2}$ and each $f_{D,T,i}$ is non-constant.
  As $\width \ge \log\log(1/\eps)$, again by \Cref{claim:long-prod-small-bias}, we have $\norm{\E[f_{D,T}(U)]}_\op \le \eps$.
  By our assumption, we have $\norm{\E[f_{D,T}(P_1)]}_{\op} \le \norm{\E[f_{D,T}(U)]}_\op + \eps \le 2\eps$.
  So altogether we have $\norm{\E[f(U)] - \E[f(G(U))]}_\op \le O(\eps)$.
  The seed length follows from the construction.
\end{proof}

\paragraph{Proof of \Cref{lem:1bit-product}.}
To get some intuition, think of $\theta$ as a constant.
Recall that the number of functions $\ell$ is roughly between $2^{30w}$ and $2^{45w}$, and $T$ is keeping each bit free with probability $p = 2^{-25w}$.
Therefore, under a typical restriction, we expect for most functions in the product, only 1 bit is set to free, and very few functions have 2 free bits.

We first need to lower-bound the probability that a non-constant function remains non-constant under a random restriction.

\begin{claim} \label{claim:restrict}
  Let $g$ be a non-constant function on $w$ bits.
  For $p \in [0,1]$, let $T$ be the distribution on $\zo^\width$, where the coordinates $T_i$'s are independent and $\E[T_i] = p$ for each $i \in [\width]$.
  With probability at least $p \cdot ((1-p)/2)^{\width-1}$ the function $g_{U,T}(x) := g(U + T \wedge x)$ is a non-constant function on $1$ bit.
\end{claim}

\begin{proof}
  Since $g$ is non-constant, there is an $x \in \zo^\width$ and a coordinate $j \in [\width]$ such that $g(x + e_j) \ne g(x)$.
  The probability that only the coordinate $T_j$ is 1 (and the rest are 0), and $U$ agrees with $x$ on the rest of the $\width-1$ coordinates is
  \begin{align*}
    \Pr\Bigl[ \bigl( T = \{j\} \bigr) \wedge \bigwedge_{i \ne j} U_i = x_i \Bigr]
    &= \Pr\Bigl[ T = \{j\} \Bigr] \cdot \Pr\Bigl[ \bigwedge_{i \ne j} U_i = x_i \Bigr] \\
    &= p \cdot (1-p)^{\width-1} \cdot 2^{-(\width-1)}
    = p \cdot \left(\frac{1-p}{2}\right)^{\width-1} . \qedhere
  \end{align*}
\end{proof}

We will use the following standard tail bound for almost $k$-wise independent random variables.

\begin{lemma}[Lemma 8.1 in \cite{LeeV-rop}] \label{lem:k-wise-tail}
  Let $X_1, \ldots, X_\terms$ be $\gamma$-almost $t$-wise independent random variables supported on $[0,1]$.
  Let $X := \sum_i X_i$, and $\mu := \E[X]$.
  We have
  \[
    \Pr\Bigl[ \abs{X - \mu} \ge \mu/2 \Bigr] \le O\left( \frac{t}{\mu} \right)^t + O\left(\frac{\terms}{\mu}\right)^t \gamma.
  \]
\end{lemma}

\begin{proof}[Proof of \Cref{lem:1bit-product}]
  We will show that for most choices of $D$ and $T$, at least $2^{3w} \theta^{-2}$ of the $f_{D,T,i}$ depend on only 1 coordinate, and the ones that depend on at least $2$ coordinates together form a $\log(1/\eps)$-junta.

  \medskip
  
  We first consider the set of functions $f_{D,T,i}$ that are restricted to $1$-bit non-constant functions.
  Let
  \[
    J_1 := \{ i \in [\terms]: \text{$\abs{T \cap I_i}=1$ and $f_{T,D,i}$ is non-constant} \} .
  \]
  If $D$ and $T$ were \emph{exactly} independent instead of almost-independent, then applying \Cref{claim:restrict} with our choice of $p \ge 2^{-23\width}\theta^3$,  we would have
  \[
    \E_{D,T}\bigl[ \abs{J_1} \bigr]
    \ge \terms \cdot p \cdot \left(\frac{1 - p}{2}\right)^{3\width-1}
    \ge ( 2^{30\width} \theta^{-5} ) \cdot (2^{-23\width} \theta^3) \cdot 2^{-3\width}
    \ge 2^{4\width} \theta^{-2} .
  \]
  As $(D,T)$ is $\delta$-almost $k$-wise independent and $\abs{I_i} \le 3\width$ for $i \in [\terms]$, the indicators $\Id(i \in J_1) : i \in [\terms]$ are $\delta$-almost $\floor{k/(3\width)}$-wise independent.
  So applying \Cref{lem:k-wise-tail} with $t = \frac{C(\log(1/\eps) + w)}{300 \width} \le \floor{k/(3w)}$ and $\gamma = \delta = \theta^k$, and recalling $k = C(\log(1/\eps) + w)$, $\terms \le 2^{45\width} \theta^{-5}$, and $\width \ge \log\log(1/\eps) + \log(1/\theta)$, we have
  \begin{align} \label{eq:j1}
    \Pr_{D,T}\bigl[ \abs{J_1} \le 2^{3\width} \theta^{-2} \bigr]
    &\le O\left( \frac{t}{2^{4\width} \theta^{-2}} \right)^{t} + O\left(2^{41\width} \theta^{-3}\right)^t \cdot \theta^k \nonumber \\
    &\le 2^{-\Omega\bigl(w \cdot \frac{C(\log(1/\eps) + w)}{300\width}\bigr)} + \theta^{k/2} \nonumber \\
    &\le \eps .
  \end{align}
  
  \medskip
  
  We now consider the $f_{D,T,i}$'s that depend on at least two coordinates.
  We will show that these functions altogether depend on at most $\log(1/\eps)$ coordinates.
  As a result, we can think of these functions as a single $\log(1/\eps)$-junta.

  \smallskip
  
  Let $J_{\ge 2} := \{ i \in [\terms]: \abs{I_i \cap T} \ge 2 \}$ be the set of functions $f_{D,T,i}$'s that depend on at least $2$ coordinates, and $Q := \bigcup_{i \in J_{\ge 2}} I_i \cap T$ be the collection of coordinates these functions depend on.
  Suppose $\abs{Q} \ge \log(1/\eps)$. 
  Then as $\abs{I_i \cap T} \ge 2$ for $i \in J_{\ge 2}$, it must be the case that some $u \le \ceil{\frac{\log(1/\eps)}{2}}$ of the subsets $I_i \cap T: i \in J_{\ge 2}$ together contain at least $2u$ many free coordinates.
  The probability of the latter event is at most
  \[
    \binom{\terms}{u} \cdot \binom{u \cdot 3\width}{2u} \cdot \Bigl (p^{2u} + \delta \Bigr) .
  \]
  Setting $u = \frac{\log(1/\eps)}{2\width} + 1$, and recalling $\terms \le 2^{45\width} \theta^{-5}$, $p = 2^{-23\width}\theta^3$ and $\delta = \theta^k \le p^{2u}$, the above is at most
  \begin{align}
    \terms^u \cdot (6\width)^{2u} \cdot 2 p^{2u}
    &\le (2^{45\width u} \theta^{5u} ) \cdot 2^{3 u \log \width} \cdot (2 \cdot 2^{-46w} \theta^6) \nonumber \\
    &\le (2^{-2\width} \theta)^u \le \eps . \label{eq:j2}
  \end{align}
  
  Let $I_0' := (T \cap I_0) \cup Q$.
  By \cref{eq:j1,eq:j2}, with probability $1-2\eps$ over $(D,T)$, we have $\abs{J_1} \ge 2^{3\width} \theta^{-2}$ and $\abs{I_0'} \le 3 \log(1/\eps)$.
  In this case, the function $f_{D,T}$ is a product of at least $2^{3\width}\theta^{-2}$ non-constant $1$-bit functions and a $(3\log(1/\eps))$-junta.
  In other words, $f_{D,T} = \prod_i f_{D,T,i}$ is a $(\terms,1,3\log(1/\eps))$-product, where $\terms \ge 2^{3\width} \theta^{-2}$, and each $f_{D,T,i} : i \in [\terms]$ is non-constant.
\end{proof}

\subsubsection{Long products have small bias: Proof of \Cref{claim:long-prod-small-bias}} \label{sec:long-prod-small-bias}
In this section, we prove \Cref{claim:long-prod-small-bias}.
We start with bounding the bias of a single arbitrary non-constant function on $w$ bits.

\begin{claim}\label{claim:bounded-bias}
  Let $\calM$ be a group of matrices supported on $\calU_\theta(d)$.
  We have $\norm{\E[g(U)]}_{\op} \le 1 - 2^{-(2\width+2)}\theta^2$ for every non-constant function $g\colon \zo^\width \to \calM$.
\end{claim}
\begin{proof}
  Let $T$ be the uniform distribution on $\zo^n$.
  Note that $\E[g(U)] = \E_{U,T}[ \E_{U'} [g(U + T \wedge U') ]]$.
  Applying \Cref{claim:restrict} with $p=1/2$, with probability at least $2^{-(2\width-1)}$, the function $g_{U,T}(x) := g(U + T \wedge x)$ is a non-constant $1$-bit function.
  Suppose $M_1 =: g_{U,T}(1) \ne g_{U,T}(0) := M_0$.
  By \Cref{claim:avg-norm}, we have $\norm{(M_1 + M_0)/2}_\op = \norm{M_1 (I + M_1^{-1}M_0)/2}_\op \le 1 - \theta^2/8$.
  So
  \begin{align*}
    \norm[\big]{ \E[g(U)] }_\op
    &\le (1 - 2^{-(2\width-1)}) \cdot 1 + 2^{-(2\width-1)} \cdot \norm{(M_1 + M_0)/2}_\op \\
    &\le 1 - 2^{-(2\width-1)} \cdot \bigl(1 - (1 - \theta^2/8) \bigr) \\
    &= 1 - 2^{-(2\width+2)} \theta^2 . \qedhere
  \end{align*}
\end{proof}

\begin{proof}[Proof of \Cref{claim:long-prod-small-bias}]
  By \cref{claim:bounded-bias}, we have $\norm{\E[f_i(U)]}_\op \le 1 - 2^{-(2\width+2)}\theta^2$ for each $i \in [\terms]$.
  Hence, for $\terms \ge 2^{2\width+2} \theta^2 \log(1/\eps)$, we have
  \[
    \norm[\big]{\E[f(U)]}_\op
    \le \prod_{i \in [\terms]} \norm[\big]{\E[f_i(U)]}_\op
    \le \bigl( 1 - 2^{-(2\width+2)}\theta^2 \bigr)^\terms
    \le \exp\bigl( -\terms \cdot 2^{-(2\width+2)}\theta^2 \bigr)
    \le \eps . \qedhere
  \]
\end{proof}

\section{Width reduction of short products: Proof of \Cref{lem:reduce-short-prod}} \label{sec:short-prod}
In this section, we prove \Cref{lem:reduce-short-prod}.
Recall that $D, T$ are $\delta$-almost $k$-wise independent distributions with $\E[D_i]=1/2$ and $\E[T_i]=p$, where
\begin{align*}
  k &= C\bigl(w + \log(m/\eps)\bigr) \\
  \delta &= (m \cdot w)^{-k} \\
  p &= 2^{-C} .
\end{align*}
for a sufficiently large constant $C$.

Given $T$, say a coordinate $i \in [n]$ is \emph{fixed} if $T_i = 0$ and is \emph{free} if $T_i = 1$.
Let $f(x) := \prod_{i=0}^\terms f_i(x_{I_i})$ be a $(\terms, 3\width, 2\log(1/\eps))$-product, where $\terms = m^5 \cdot 2^{30\width}$.
We will show that with high probability over $T$, (1) at most $\log(1/\eps)$ of the $2\log(1/\eps)$ coordinates in $I_0$ are free; (2) for most of the $I_i: i \ge 1$, at most $2\width$ coordinates in each of them are free, and (3) in the remaining $I_i$'s, there are at most $\log(1/\eps)$ many free coordinates in total.

\smallskip
To proceed, let 
\[
  J := \{ i \in [\terms]: \abs{T \cap I_i} \ge 2\width \} \quad\text{and}\quad  Q := \bigcup_{j \in J} I_j \cap T.
\]
It suffices to show the following two claims.
\begin{claim} \label{claim:reduce-input-length}
  $\abs{Q} \le \log(1/\eps)$ with probability $1 - \eps$ over $T$.
\end{claim}

\begin{claim} \label{claim:reduce-junta-size}
  $\abs{T \cap I_0} \le \log(1/\eps)$ with probability $1 - \eps$ over $T$.
\end{claim}

\begin{proof}[Proof of \Cref{lem:reduce-short-prod}]
	Let $I_0' = (T \cap I_0) \cup Q$.
	By \Cref{claim:reduce-input-length,claim:reduce-junta-size}, with probability $1 - 2\eps$ over $T$, we have $\abs{I_0'} \le 2\log(1/\eps)$, and for every $i \in [\terms]\setminus J$, we have $\abs{T \cap I_i} \le 2\width$.
	Therefore, the function $f_{D,T}$ is a $(\terms,2\width,2\log(1/\eps))$-product.
\end{proof}

\begin{proof}[Proof of \Cref{claim:reduce-input-length}]
  Suppose $\abs{Q} \ge \log(1/\eps)$.
  Then as $\abs{T \cap I_j} \ge 2\width$ for $j \in J$, it must be the case that some $u \le \ceil{\frac{\log(1/\eps)}{2\width}}$ subsets $T \cap I_j: j \in J$ altogether contain $2\width \cdot u$ many free coordinates.
  This happens with probability at most
  \[
    \binom{\terms}{u} \cdot \binom{3\width \cdot u}{2\width \cdot u} \cdot \Bigl( p^{2\width \cdot u} + \delta \Bigr) .
  \]
  Setting $u = \frac{\log(1/\eps)}{3\width} + 1$, and recalling $\terms \le m^5 2^{30\width}$ and $\delta = (m \cdot w)^{-k} \le p^{2w\cdot u}$, the above is at most
	\begin{align*}
    \terms^u \cdot 2^{3\width \cdot  u} \cdot 2 p^{2\width \cdot u}
    &\le \bigl( m^5 2^{30\width} \bigr)^u \cdot 2^{3\width \cdot u} \cdot 2^{-C \cdot  2\width \cdot u + 1} \\
    &\le 2^{u \bigl( 33\width + 5 \log m - C \cdot 2\width \bigr)} \\
		&\le 2^{-u \cdot 3\width}
		\le \eps 
	\end{align*}
  where in the second last inequality we used $\width \ge \log m$.
\end{proof}

\begin{proof}[Proof of \Cref{claim:reduce-junta-size}]
  Recall that $\abs{I_0} \le 2\log(1/\eps)$, and $\delta = (m \cdot w)^{-k} \le p^{\log(1/\eps)}$.
  So
  \begin{align*}
    \Pr\Bigl[ \abs{T \cap I_0} \ge \log(1/\eps) \Bigr]
    &\le \binom{\abs{I_0}}{\log(1/\eps)} \Bigl( p^{\log(1/\eps)} + \delta \Bigr) \\
    &\le 2^{2\log(1/\eps)} \cdot 2^{-C\log(1/\eps) + 1} \\
    &\le \eps/2 . \qedhere
  \end{align*}
\end{proof}

\section{Fooling $(1,w,3\log(1/\eps))$-products over groups} \label{sec:prg-spill-product}

In this section, we show how to extend the PRGs for $(\ell,1,0)$-products over $p$-groups (\Cref{thm:fooling-p-groups}) and commutative groups~\cite{GopalanKM15} to fool $(\ell,1,3\log(1/\eps))$-products.

\subsection{$p$-groups}

We use the fact that our generator in \Cref{thm:fooling-p-groups} is simply the XOR of independent copies of small-bias distributions.
The following claim shows that conditioning on a small number of bits of a small-bias distribution remains small bias.

\begin{claim} \label{claim:bias-conditioning}
  Let $D$ be an $\eps$-biased distribution on $\zo^n$.
  For any set $S$ and $y \in \zo^S$, the distribution of $D$ conditioned on $D_S = y$ is $(2^{\abs{S}+1} \eps)$-biased on $\zo^{[n] \setminus S}$.
\end{claim}
\begin{proof}
  We may assume $\eps \le 2^{-(\abs{S}+1)}$, for otherwise the claim is vacuous.
  For a subset $T \subseteq [n]$, let $\chi_T(x) := (-1)^{\sum_{i\in T}x_i}$ be any parity test.
  Let $T$ be any nonempty subset of $[n]\setminus S$.
  First observe that
  \[
    \Id(D_S = y)
    = \prod_{i \in S} \frac{1 - \chi_{\{i\}}(D)}{2}
    = 2^{-\abs{S}} \sum_{S' \subseteq S} (-1)^{\abs{S'}} \chi_{S'}(D) .
  \]
  Taking expectations on both sides and applying the triangle inequality, we have $\Pr[D_S = y] \ge 2^{-\abs{S}} - \eps \ge 2^{-(\abs{S} + 1)}$.
  Note that
  \begin{align*}
    \E[\chi_T(D) \mid D_S = y] \Pr[ D_S = y]
    &=\E[\chi_T(D) \cdot \Id(D_S = y)] \\
    &=\E\Bigl[\chi_T(D) \cdot \prod_{i \in S} \frac{1 - \chi_{\{i\}}(D)}{2}\Bigr] \\
    &= 2^{-\abs{S}} \sum_{S' \subseteq S} (-1)^{\abs{S'}} \E\bigl[\chi_{T \cup S'}(D)\bigr] .
  \end{align*}
  So its magnitude is bounded by $\eps$.
  Therefore, $\abs[\big]{\E[\chi_T(D) \mid D_S = y]} \le \Pr[ D_S = y]^{-1} \eps \le 2^{\abs{S}+1} \eps$.
\end{proof}

\begin{corollary} \label{cor:prg-p-groups}
  There is a PRG that $\eps$-fools $(n,1,3\log(1/\eps))$-products over any $p$-groups of order $m$  with seed length $O_m(\log(n/\eps))$.
\end{corollary}
\begin{proof}
  Recall our generator for $p$-groups in \Cref{thm:fooling-p-groups} is simply the XOR of independent copies of $(\eps/n)^c$-biased distributions.
  By \Cref{claim:bias-conditioning}, for any fixing of the input bits of $f_0$ in each copy, each distribution remains $(2/\eps)(\eps/n)^c$-biased.
\end{proof}

\subsection{Commutative groups}

We now show that the Gopalan--Kane--Meka PRG fools $(\ell,1,3\log(1/\eps))$-products.
We will use the following PRG by Gopalan, Kane, and Meka~\cite{GKM18} that fools $(\ell,1,0)$-products over $\C$.
A simple argument shows that the same PRG also fools $(\ell,1, 3\log(1/\eps))$-products.

\begin{claim} \label{claim:gkm-prg}
  There is an explicit PRG $P$ that $\eps$-fools $(\ell,1,3\log(1/\eps))$-products over commutative groups of order $m$ with seed length $O_m(\log(\ell/\eps)(\log\log(\ell/\eps)^2)$.
\end{claim}

\begin{lemma}[Theorem 1.1 and Lemma 9.1 in \cite{GKM18}] \label{lem:gkm-prg}
  There is an explicit $P_{\GKM}\colon \zo^s \to \zo^n$ where $s = O(\log(\ell/\eps))(\log\log(\ell/\eps))^2$ such that the following holds.
  If $w \in \R^n$ satisfies $\sum_i \abs{w_i} \le W$, then
  \[
    \dist_{TV}\bigl(\angles{w,U}, \angles{w,P_{\GKM}(U)}\bigr) \le O(\sqrt{W}) \cdot \eps . \qedhere
  \]
\end{lemma}

\begin{proof}[Proof of \Cref{claim:gkm-prg}]
  By \Cref{claim:closeness}, it suffices to fool the product over each irreducible representation of $G$ with error $\eps/\sqrt{m}$.
  Since $G$ is commutative, all its irreps are 1-dimensional.
  Moreover, they are supported on subsets of $\C_m := \{z \in \C: \abs{z}^m = 1\}$.
  Let $f = \prod_{i=0}^\ell f_i$ be an $(\ell,1,3\log(1/\eps))$-product over $\C_m$.

  Let $\omega := e^{i\frac{2\pi}{m}}$.
  Note that for any 1-bit function $g\colon \zo \to \C_m$ we can write $g$ as
  \[
    g(y)
    = \omega^{a_1 y + a_0 (1-y)}
    = \omega^{a_0} \cdot \omega^{(a_1 - a_0) y}
  \]
  for some $a_0, a_1 \in \{0, \ldots, m-1\}$.
  We can also write $f_0(x_{I_0})$ as $\omega^{h(x_{I_0})}$, for some $h\colon \zo^{I_0} \to \{0, \ldots, m-1\}$.
  Therefore, $f$ has the form of
  \[
    f(x) = \omega^b \cdot \omega^{\sum_{j \in J} a_j x_j + h(x_{I_0})} ,
  \]
  for some coefficients $b$ and $a_j$'s taking values from $\{0, \ldots, m-1\}$, and $J$ and $I_0$ are disjoint subsets of $[n]$.
  Write $I_0 = \{r_1 < \cdots < r_{\abs{I_0}}\}$ for some $r_j \in [n]$.
  Consider the integer-valued function $F\colon\zo^\ell \to \Z$ defined by
  \[
    F(x) := \sum_{j \in J} a_j x_{j} + 2^{\ceil{\log(m\ell)}} \sum_{j=1}^{\abs{I_0}} 2^{j-1} x_{r_j} .
  \]
  So the first $\ceil{\log m\ell}$ bits of $F(x)$ encode the first sum, and the last $\abs{I_0}$ bits are the decimal encoding of the binary string $x_{I_0}$.
  Note that we can compute $f(x)$ given $F(x)$.
  Moreover, $F(x) = \angles{w,x}$ for some $w \in \R^n$ with $\sum_i \abs{w_i} \le O(m\ell/\eps^3)$.
  Therefore, if we let $P$ be the PRG in \Cref{lem:gkm-prg} with error $O(\frac{\eps^3}{m \ell})$, which uses a seed of $O(\log(m\ell/\eps))(\log\log(m\ell/\eps))^2$ bits, then it follows that $\dist_{TV}(F(U),F(P_{\GKM}(U)) \le \eps/\sqrt{m}$.
\end{proof}

\section{Mixing characterization of Dedekind groups} \label{sec:proof-dedeking-mixing}
In this section we give a proof of \Cref{lem:dedekind-mixing}, which says that a (finite) group $G$ is mixing if and only if it is Dedekind.
This proof is provided by Yves de Cornulier on~\url{https://mathoverflow.net/a/482286/8271}.

Recall from \Cref{def:mixing} that a (finite) group $G$ is \emph{mixing} if for every nontrivial irreducible (unitary) representation $\rho$ and non-identity element $g \in G$, the matrix $\rho(g)$ has no eigenvalue 1, and $G$ is Dedekind if it has the form $\mathbb{Q}_{8}\times\mathbb{Z}_{2}^{t} \times D$ for any integer $t$ and commutative group $D$ of odd order.

We first note that the definition of $G$ being mixing is equivalent to the following condition:  for every $g \in G$ and irrep $\rho$, the subspace $\ker(\rho(g)-I)$ is a subrepresentation.
The proof of \Cref{lem:dedekind-mixing} follows from the following two claims.
Here we use the equivalence that a group $G$ is Dedekind if and only if every subgroup of $G$ is normal.

\begin{claim}
  If $G$ is Dedekind, then $\ker(\rho(g) - I)$ is a subrepresentation for every $g$ and $\rho$.
\end{claim}
\begin{proof}
  Take an element $g \in G$.
  By definition of Dedekind, every subgroup in $G$ is normal.
  In particular $\angles{g}$ is also normal.

  Take $v \in W_g := \{v : \rho(g) v = v\}$.
  To show that $W_g$ is a subrepresentation, we need to show that $\rho(g) (\rho(h) v) = \rho(h) v$ for every $h$.
  But this is equivalent to showing 
  \[
    \rho(h)^{-1} \rho(g) \rho(h) v
    = \rho(h^{-1} g h) v
    = v .
  \]
  Since $\angles{g}$ is normal, we have $h^{-1} g h = g^i$ for some $i$.
  It is clear that $\rho(g^i)v  = \rho(g)^i v = \rho(g)^{i-1} (\rho(g) v) = \rho(g)^{i-1} v = \cdots = v$.
  So indeed $W_g$ is a subrepresentation.
\end{proof}

\begin{claim}
  If $\ker(\rho(g) - I)$ is a subrepresentation for every $g$ and $\rho$ of $G$, then $G$ is Dedekind.
\end{claim}

\begin{proof}
  As in the previous claim, we consider $W_g = \ker(\rho(g) - I) = \{v : \rho(g) v = v\}$.
  Note that for $v \in W_g$, we have $\rho(g^i) v = \rho(g)^i v = v$.
  Suppose $W_g$ is a subrepresentation.
  Take $h \in G$ and $v \in  W_g$.
  Since $\rho(h) W_g \subseteq W_g$, we have $\rho(g) \rho(h) v = \rho(h) v$.
  That means $\rho(h^{-1} g h) v = v$ for every $v \in W_g$.
  This implies $\angles{h^{-1}g h} = \angles{g}$, meaning $h^{-1}g h \in \angles{g}$ and thus $\angles{g}$ is normal.
\end{proof}

\subsection*{Acknowledgements}
We thank Yves de Cornulier for answering a question about Dedekind groups and providing a proof of \Cref{lem:dedekind-mixing}.

\bibliographystyle{alpha}
\bibliography{OmniBib,ref}

\appendix

\end{document}